%% file: main.tex
\documentclass[11pt, runningheads]{llncs}
\usepackage[T1]{fontenc}
\usepackage{graphicx}
\usepackage{soul}
\usepackage[linesnumbered,ruled,vlined]{algorithm2e}

\usepackage[margin=1.2in]{geometry}

\input{includes-and-macros-sagt2025}

\usepackage{todonotes}
\usepackage{cleveref}

\newcommand{\dkb}[1]{\textcolor{violet}{#1}}
\newcommand{\dinesh}[1]{\dkb{(Dinesh says: #1)}}

\usepackage[most]{tcolorbox}
\makeatletter
\newcommand*{\rom}[1]{\expandafter\@slowromancap\romannumeral #1@}
\makeatother

\begin{document}
\title{Fairly Wired: 
Towards Leximin-Optimal Division of Electricity}
%
%
%

\author{Eden Hartman\thanks{Eden Hartman and Dinesh Kumar Baghel 
 contributed equally to this paper.}\inst{1}\orcidID{0000-0001-5819-4432} \and
Dinesh Kumar Baghel
\inst{2}\orcidID{0000-0001-8518-7908} \and
Erel Segal-Halevi\inst{2}\orcidID{0000-0002-7497-5834}}
\institute{Bar-Ilan University, Israel 
\email{eden.r.hartman@gmail.com}
\and
Ariel University, Israel
\email{\{dinkubag21,erelsgl\}@gmail.com}
}

\maketitle      

\begin{abstract}
    In many parts of the world—particularly in developing countries—the demand for electricity exceeds the available supply.
    In such cases, it is impossible to provide electricity to all households simultaneously. This raises a fundamental question: how should electricity be allocated fairly?
    In this paper, we explore this question through the lens of egalitarianism—a principle that emphasizes equality by prioritizing the welfare of the worst-off households. One natural rule that aligns with this principle is to maximize the egalitarian welfare—the smallest utility across all households. 
    We show that computing such an allocation is NP-hard, even under strong simplifying assumptions.

    Leximin is a stronger fairness notion that generalizes the egalitarian welfare: it also requires to maximize the smallest utility, but then, subject to that, the second-smallest, then the third, and so on.
    The hardness results extends directly to leximin as well.
    Despite this, we present a Fully Polynomial-Time Approximation Scheme (FPTAS) for leximin in the special case where the network connectivity graph is a tree.
    This means that we can efficiently approximate leximin—and, in particular, the egalitarian welfare—to any desired level of accuracy.

    
\end{abstract}

%
%
%
\section{Introduction}\label{sec: Introduction}
Electricity is a public resource crucial for enhancing people's lives and fostering the development of both social and economic sectors. 
In principle, as electricity owned collectively by the citizens of a country, it ought to be distributed equally among all.
However, such distribution is not always possible as there are places in the world in which the demand is higher than the available supply. 
In Nigeria, for instance, $60-70\%$ of the population do not have access to electricity \cite{Oyedepo_2012}.
To compare, in 2022, the total energy production in the United Kingdom is $\approx 324K$ Gigawatt-hour  \cite{IEA_UK_Electricity_2022} for a population of approximately $67.6$M \cite{StatistaUKPopulation2022}; whereas the electricity production in Nigeria is  
$\approx 38K$ Gigawatt-hour 
 \cite{IEA_Nigeria_Electricity_2022} for a population of approximately $216.7$M people \cite{StatistaNigeriaPopulation2022}. 
Unfortunately, it is estimated that the problem is going to get even worst~\cite{Kaygusuz_2012}.

When facing this gap in electricity demand and supply, it is often impossible to supply electricity to all households at the same time.
In such scenario, authorities generally disconnect a significant part of their regions from the supply. This disconnection of a portion of region from the supply is known as \emph{load shedding}. 
The only objective of load shedding is to maintain the balance between demand and supply, and it usually done without considering fair aspects. 

Can the authorities fairly distribute electricity while maintaining the balance? 
This is the problem we study in this paper. 
Our goal is to develop a mechanism for distributing electricity between households in a fair manner.
That is, given the \emph{supply} (the amount of electricity available), the \emph{demand} (amount of electricity each household requires), and a graph describing the households' connections to the power station and to each other,
the mechanism should determine when each family should be connected, together with whom, and for how long.


%
%
We focus on the egalitarianism principle, which promotes equality by prioritizing the welfare of the least advantaged. Under this principle, we study two natural fairness rules.
The egalitarian rule, which selects a solution that maximizes the minimum utility; and
the leximin rule, which generalizes the egalitarian rule by first maximizes the minimum utility, then subject to that, the second-minimum, then the third, and so on.

%

We aim to take into account Geographic considerations. That is, we assume that there is a fixed graph describing the electricity distribution network, and the set of households receiving electricity at each point in time, together with the power station, must form a connected sub-graph of it.

We assume that the utility of an agent is positive only if the agent receives the requested demand; otherwise the utility is as good as receiving nothing. For example, at some time a household is running a washing machine which requires $1kW$ of electricity to function so a partial fulfillment of this demand will prohibit household to perform the activity.

\paragraph{Contributions.}
We show that, without connectivity constraints, it is NP-hard to compute an \emph{egalitarian allocation} of electricity --- an allocation that maximizes the smallest utility. As any leximin allocation is egalitarian  by definition, it is clearly NP-hard too\Cref{sec:hardness}.

Our main contribution is an FPTAS for the leximin allocation problem, for the case in which the electricity network is a tree. 
We present an algorithm that, given any vector of demands, any tree network and any $\epsilon\in (0,1)$, obtains an $(1-\epsilon)$-approximation to the leximin-optimal allocation (see \Cref{sec:model} for the exact definition), and runs in time polynomial in the problem size and in $1/\epsilon$ (\Cref{sec:fptas}).

The main technical tool we use
is a generic reduction from the problem of finding a (fractional) leximin-optimal allocation to finding a \emph{weighted utilitarian-optimal allocation} --- an allocation that maximizes a weighted sum of the agents' utilities 
\cite{hartman2025reducing}.
In our setting, computing a weighted utilitarian allocation  is equivalent to solving a variant of the Knapsack problem with particular connectivity constraints, which we call \emph{Geographic Knapsack}. This variant, like the original Knapsack problem, is NP-hard; we present an FPTAS for the case of a tree graph (\Cref{alg:SolveGeographicKnapsack}).
The reduction in \cite{hartman2025reducing} is robust in the sense that, given an approximation algorithm to the weighted utilitarian-optimal problem, it yields an approximation algorithm to the leximin-problem, with the same approximation factor; hence, our FPTAS for Geographic Knapsack yields the desired FPTAS for leximin-optimal electricity allocation.

\subsection{Related Work}
\label{sec: Related Literature}



\paragraph{A high level overview of electricity grid:}\footnote{Adapted from \cite{electricity-grid}.} Traditionally, electricity distribution grid is organized hierarchically. The power plant is at the top of this hierarchy. The power flows from the top to the transmission substation from which it flows further to the distribution grid. From there, power flows in different directions to distribution buses, and from there to individual households. Lines connect households to the bus. 

We did not find any literature on electricity distribution with the same Geographic constraint mentioned in this paper. So in this section, we are only considering the relevant literature at the bus level and at the household level.

\paragraph{Load shedding at bus level:}The authors in \cite{Pahwa_Scoglio_Das_Schulz_2013} presented three load shedding strategies at the bus level: homogeneous load shedding strategy, linear optimization, and tree load reduction heuristic. The proposed tree heuristics resolve the issue of other load shedding strategies. In tree load reduction heuristic, a tree is formed using the initially failed line, and then the tree heuristic reduces the same percentage of load from a subset of lines in the tree. They argued that tree heuristics do not work for all failure cases in the system. 
Shi and Liu \cite{Shi_Liu_2015}, have given a decentralized load shedding mechanism to prevent cascading failures in a smart electric grid. 
In their solution, the nodes (buses) communicate with each other and determine the amount of load to be shed. The nodes are compensated in real-time to cover the shedding costs. Implementing their mechanism requires an advanced information and communication technologies infrastructure. However, developing countries require a simple solution.

\paragraph{Partial allocation of demand:}Buerman et. al \cite{Buermann_Gerding_Rastegari_2020} has studied fair allocation with stochastic resource availability. In their work, an household derives positive utility with the partial fulfillment of their demand, whereas in our case, 
households derive positive utility when the requested demand is fulfilled; otherwise, their utility is zero. 
Janjua et al \cite{Janjua_Ali_Kallu_Ibrahim_Zafar_Kim_2021},  
compared game-theoretic bankruptcy and Nash bargaining based solutions 
at the province level. Their solutions allocate a fraction of the demand that may prohibit households to carry out critical work.
In another direction, the researchers in \cite{Akasiadis_Chalkiadakis_2017,Ali_Mansoor_Khan_Arshad_Faizullah_Khan_2021} have provided solutions to the load shedding problem in which households can shift their demand to some other time of the day. In practice, there are many activities that cannot be shifted to some other time of the day, and doing so will impact the comfort of an agent. 
The Gerding et al. \cite{Gerding_Perez-Diaz_Aziz_Gaspers_Marcu_Mattei_Walsh_2019} deals with allocating limited electricity to EVs in an online setting, where not all EVs may receive their full requested charge.
However, in our model, the fractional fulfillment of an agenF's demand is of no use.
Stein et al. \cite{Stein_Gerding_Robu_Jennings_2012}, presented a model-based online mechanism for pure electric vehicles. In their setting, an agent values a particular amount of resource, and do not derive additional value for receive more and values it to $0$ if receives less.  In contrast, our model allows agents to derive utilities from receiving varying amounts of the resource.
At a more granular level, Azasoo et al. \cite{Azasoo_Kanakis_Al-Sherbaz_Agyeman_2019} have presented a solution that meets a fraction of the agent’s demand, which is further distributed among the high priority tasks using the smart-meter (which are less prevalent in developing countries). If all tasks are equally important, then their solution will prevent agents to perform some essential activities.

\paragraph{Household level load shedding:}The above work and other current approaches are focused on load shedding at the bus level and are often unfair at the household level. Household level load shedding solutions has been discussed in \cite{Oluwasuji_Malik_Zhang_Ramchurn_2018} in which the authors have presented four heuristic algorithms. Their heuristic algorithms failed to consider the heterogeneous demands of the households and therefore resulted in disconnecting households from supply when they needed it the most. Later, they improved upon their result in \cite{Oluwasuji_Malik_Zhang_Ramchurn_2020} in which they modeled the household load shedding problem as a multiple knapsack problem (MKP) and solved it using Integer Linear Programming (ILP). 

Recently, Baghel, Ravsky and Segal-Halevi \cite{Baghel_Ravsky_Segal-Halevi_2024,Baghel_Ravsky_Segal-Halevi_2025} presented bin-packing based solutions for the fair allocation of electricity. They focused on the egalitarian allocation of connection time. They presented both theoretic approximation algorithms, and practical heuristic algorithms.

We improve over their work in two ways: first, we find an (approximate) leximin electricity allocation, which is stronger than egalitarian; second, we handle connectivity constraints assuming the distribution network has a tree structure.

Other relevant work has been incorporated into the section's discussion.

\section{Model}\label{sec:model} 

An instance of the electricity division problem is defined as follows.
The time during which electricity is available is represented by the interval $[0,1]$, which may correspond to any relevant period—such as a day, an hour, or a year (we normalize to the unit interval for simplicity).
At each point in time, there is a fixed supply of electricity, denoted by $S \in \mathbb{R}_+$ (measured in kilowatts), which is assumed to remain constant throughout the entire interval.

The set of agents is denoted by $N:=\{1,\ldots,n\}$, where each agent $i \in N$ represents a household or a family. 
Each agent $i \in N$ has a fixed demand $d_i \in \mathbb{R}_{> 0}$, representing the exact amount of electricity required to connect their household. Receiving less than $d_i$ is equivalent to receiving nothing, and receiving more provides no additional benefit.
We assume that the demands are determined externally rather than reported by the agents. As a result, there is no strategic aspect related to truth-telling or misreporting of demand.

For a subset $C\subseteq N$, we denote the total demand of agents in $C$ by $d_C := \sum_{i\in C} d_i$.
We assume that the total demand is higher than the supply, that is, $d_N > S$.

The network through which electricity is transmitted is modeled by a graph $G = (V, E)$, where $V := N \cup \{s\}$ includes one vertex for each agent and a special vertex $s$ representing the power station; and there is an edge between two vertices if there is a power line connecting them. 

\paragraph{Configurations.} We know that it is impossible to connect all agents to electricity at the same time, but it is possible for some subsets of the agents.
To capture this, we refer to a subset of agents, $C \subseteq N$,  as a \emph{configuration}, and say that the configuration $C$ is \emph{feasible} if it satisfies the following two requirements:
\begin{itemize}
    \item Demand Constraint: the total demand of the agents in $C$ does not exceed the available supply: $d_C \leq S$. 
    
    \item Geographic Constraint: the subgraph induced by the agents in $C$ together with the power station is connected. 
    That is, let $V' := C \cup \{s\}$ and $E'$ be the set of edges in $G$ connecting nodes in $V'$, then the subgraph $G' := (V', E')$ must be connected.
\end{itemize}
We denote the set of all feasible configurations, given the demand vector $\mathbf{d}$ and the supply $S$ and the graph $G$, by $\mathcal{C}(\mathbf{d},S,G)$, or simply $\mathcal{C}$ when the arguments are clear from the context.

\paragraph{Solutions.} A \emph{solution} to the electricity division problem is an allocation $\mathbf{A} = (A_1, \ldots, A_n)$, where each $A_i$ is a union of finitely many disjoint sub-intervals of $[0,1]$, representing the times during which agent $i$ is connected to the power supply. 
Notice that the $A_i$ need not be disjoint, as it is possible to connect several agents at the same time, as long as their total demand is at most $S$.%
\footnote{
This is in contrast to the classic problem known as \emph{cake-cutting}, in which an interval is partitioned among $n$ agents such that the agents' pieces must be pairwise-disjoint.
}
A solution is \emph{feasible} if, for every interval $I \subseteq [0,1]$, the set of agents who hold this piece, $\{i \in N \mid I \subseteq A_i\}$, is a feasible configuration.
The set of all feasible solution is denoted by $\mathcal{A}$.

\paragraph{Utilities.} 
Each agent $i$ has a utility function, $u_i \colon \mathcal{A} \to \mathbb{R}_{\geq 0}$, that describes the satisfaction level from a given allocation. 
In this paper, we assume that the agents' utilities are \emph{uniform} and \emph{additive} --- the agents care only about the total duration of connection. That is, $u_i(\mathbf{A}) = \sum_{I \in A_i} |I|$ for any $i \in N$ (where $A_i$ is the union of disjoint sub-intervals given to $i$).

Although this assumption may seem overly restrictive and not fully reflective of real-world conditions, we prove in the next section that the problem is NP-hard even under such utilities.
Further, this assumption is becoming more reasonable in many developing regions, where home-batteries are gradually becoming more common. These allow households to store electricity and use it throughout the allocation period (e.g., a full day), making the total duration of connection more important than its specific timing.

\paragraph{Example.}
Assume there are three agents $N:= \{1,2,3\}$, each with demand of $2$ ($d_i =2$ for any $i \in N$). The supply is $S=4$. The graph $G$ contains only two paths: $s-1-2$ and $s-3$.
Notice that indeed $\sum_i d_i > S$.

The feasible configurations are $\{1\}, \{1,2\}, \{3\}$ and $\{1,3\}$.
The configuration $\{1,2,3\}$ is not feasible since it does not satisfy the demand constraint.
The configurations $\{2\}$ and $\{2,3\}$ are not feasible since they do not satisfy the Geographic constraint. 

Consider the following allocation: $A_1:= \{[0,1]\}$, $A_2:= \{[0,1/2]\}$ and $A_3:= \{[0,1/2]\}$.
The allocation is feasible since each set of agents who hold the same interval is a feasible configuration. Specifically, $\{1,2\}$ hold the interval $[0,1/2]$ and $\{1,3\}$ hold the interval $[1/2,1]$.
The utilities of the agents from this allocation are: $u_1(\mathbf{A}) = 1$, $u_2(\mathbf{A}) = 1/2$ and $u_3(\mathbf{A}) = 1/2$.

In contrast, consider the allocation: $A_1:= \{[0,1/3], [2/3,1]\}$, $A_2:= \{[0,2/3]\}$ and $A_3:= [1/3,1]$.
This allocation is \emph{not} feasible since the set of agents $\{2,3\}$ hold the interval $[1/3,2/3]$, but $\{2,3\}$ is not a feasible configuration.

\subsection{Fairness Measurements}
In this paper, we focus on egalitarianism---a principle that aims to attain fairness by improving the welfare of the worst-off agents.  
Specifically, we consider \emph{egalitarian optimization} and its refinement, \emph{leximin optimization}.

\paragraph{Egalitarian Optimal.} An egalitarian optimal allocation, $\mathbf{A}^{eg}$, is one that maximizes the smallest utility:
\begin{align*}
    \min_{i \in N} u_i(\mathbf{A}^{eg}) \geq \min_{i \in N} u_i(\mathbf{A}) \quad \text{ for any } \mathbf{A} \in \mathcal{A}
\end{align*}

\paragraph{Leximin Optimal.} A leximin optimal allocation is one that one that maximizes the smallest utility, and, subject to that, maximizes the second-smallest utility, and then the third, and so on.
To define it formally, we start by describing the leximin order over vectors of size $n$.

Let $\mathbf{x}\in \mathbb{R}^n$. We denote by $\mathbf{x}^{\uparrow}$ the corresponding vector sorted in non-decreasing order, and by $x^{\uparrow}_i$ the $i$-th smallest element counting repetitions.
For two vectors, $\mathbf{x}, \mathbf{y}\in \mathbb{R}^n$, 
$\mathbf{x}$ is said to be leximin-preferred over $\mathbf{y}$, denoted $\mathbf{x} \succeq \mathbf{y}$, if one of the following holds. Either $\mathbf{x}^{\uparrow} = \mathbf{y}^{\uparrow}$; or that there exists an integer $1 \leq k \leq n$ such that $x^{\uparrow}_k > y^{\uparrow}_k$ and $x^{\uparrow}_j = y^{\uparrow}_j$ for any $1 \leq j < k$. 

For a given allocation, $\mathbf{A}$, let $\mathbf{u(A)}:=(u_1(\mathbf{A}), \ldots, u_n(\mathbf{A}))$ be the vector of utilities attained by all agents. Notice that it is a vector of size $n$.
We denote by $\mathbf{u}^{\uparrow}(\mathbf{A})$ the corresponding vector sorted in non-decreasing order, and by $u^{\uparrow}_i(\mathbf{A})$ the $i$-th smallest element in the vector $\mathbf{u}^{\uparrow}(\mathbf{A})$.

A leximin optimal allocation, $\mathbf{A}^{lex}$, is one for which $\mathbf{u}(\mathbf{A}^{lex}) \succeq \mathbf{u}(\mathbf{A})$ for any $\mathbf{A} \in \mathcal{A}$.

\paragraph{Leximin Approximation.} We follow the definition of Hartman et al. \cite{hartman2025reducing}. For a given approximation factor $\alpha \in [0,1]$, an allocation is $\alpha$-leximin-approximation, $\mathbf{A}^{\alpha Lex}$, if $\mathbf{u}(\mathbf{A}^{\alpha Lex}) \succeq \alpha \cdot \mathbf{u}(\mathbf{A})$ for any $\mathbf{A} \in \mathcal{A}$.

\paragraph{FTPAS for Leximin.} An FTPAS for leximin is an algorithm that, given any constant $\epsilon\in (0,1)$, returns an allocation that is $(1-\epsilon)$-leximin-approximation, and runs in time polynomial in $1/\epsilon$ and the problem size.

\section{Computing an Egalitarian Optimal Allocation is NP-Hard}
\label{sec:hardness}

In this section, we prove that computing an egalitarian-optimal allocation of electricity is NP-hard even when there are no Geographic constraints: all agents are directly connected to the power station, meaning that the graph $G$ is a star centered at the vertex $s$. 
In this setting, any group of agents that satisfies the demand constraint can be connected at the same time.

    

Specifically, we show that determining whether the egalitarian welfare is at least $1/2$ is NP-hard, which implies that computing an egalitarian-optimal allocation is also NP-hard.

{Now, we will give proof of the problem $\mathrm{MaxEED}$: Maximize Egalitarian Electricity Distribution Problem defined as ``find the largest $r$ such that each agent can be connected at least $r$ of the time''.
}
    \begin{theorem} \label{thm:finding r is NP hard}
	$\mathrm{MaxEED}$ is NP-hard. 
    In particular, it is NP-hard to decide if it is possible to connect each household at least $r=1/2$ of the time.
\end{theorem}
\begin{proof}
	{We prove this by reduction from the \textsc{Partition} problem as follows:}
	
	Let $\mathbf{d} = (d_1, \ldots, d_n)$ be an instance of the partition problem. We can assume that $d_N$
    is even otherwise, trivially the answer is ``no''. 
    
    We construct an instance of MaxEED with supply $S = d_N/2$, and demand vector $\mathbf{d}$.

    Note that a solution to MaxEED in which each agent is connected exactly $r=1/2$ of the time always results in an allocation of supply equal to $S$. We will use this fact in our proof.

    Now we will prove that the partition problem is a ``Yes'' instance iff there is a solution to the MaxEED in which is each agent is connected at least $r=1/2$ of the time.

    \begin{itemize}
        \item Suppose there is a solution to the \textsc{Partition} instance. Denote it by $(C_1,N\setminus C_1)$. then connecting each of these two subsets for half the time will result in a feasible solution to MaxEED in which each household is connected exactly $r=1/2$ of the time.
        \item Let $C_1, \ldots C_q$ be subsets of $N$ in a solution to MaxEED. Let us assume that each subset $C_j$ is connected $t_j > 0$ of the time such that $\sum t_j = 1$ and each item is connected at least $r=1/2$ of the time.  
        Then, for each household $i$, $\sum_{C_j\ni i} t_i \geq 1/2$, so the total amount of electricity given to $i$ is at least $d_i/2$.
        Therefore, the total amount of electricity given to all households together is at least $d_N/2$, which equals $S$ by construction. But since the total available supply is $S$, the inequality must in fact be an equality, which means that each configuration must be completely full, that is,  $d_{C_j} = S$ for every configuration $C_j$. Hence $d_{N\setminus C_j}= S$ too, so $(C_i,N\setminus C_i)$ is a solution to the \textsc{Partition} instance.
    \qed
    \end{itemize}
\end{proof}

\section{An FPTAS for Leximin in Tree Electricity Networks}
\label{sec:fptas}
In this section we present our main result: a Fully Polynomial-Time Approximation Scheme (FPTAS) for leximin in the special case where the network connectivity graph is a tree. 

We start by noting that since agents care only about their total connection time, it is sufficient to specify how much time each feasible configuration is connected, rather than the exact scheduling of these intervals.
For example, allocating $[0,0.7]$ to agent $1$ and $[0.4,1]$ to agent $2$ yields the same utilities as any other allocation that also connects configuration $\{1\}$ for $0.4$ of the time, configuration $\{1,2\}$ for $0.3$ of the time, and configuration $\{2\}$ for $0.3$ of the time.
In other words, an allocation yields a given utility vector if and only if there exists a vector of connection times over feasible configurations that produces the same utilities.
Because of this equivalence, it is sufficient to consider only these vectors when searching for an optimal or approximately optimal solution.

We denote the set of all possible connection time vectors over feasible configurations by
\begin{align*}
    X := \{\mathbf{x} \in [0,1]^{|\mathcal{C}|} ~\mid ~\sum_{j=1}^{|\mathcal{C}|} x_j =1\}
\end{align*}
The utility of agent $i$ from a given configuration $C$ is $u_i(C):= \mathbf{1}_{i \in C_j}$ where $\mathbf{1}_{i \in C_j}$ equals $1$ if $i \in C_j$ and $0$ otherwise; while the utility from $\mathbf{x} \in X$ becomes $u_i(\mathbf{x})= \sum_{j=1}^{|\mathcal{C}|} x_j\cdot u_i(C_j)$. 
Notice that the number of feasible configurations may be exponential in $n$, so even efficiently describing such a solution is not straightforward. 
However, our algorithm always returns a \emph{sparse} solution—one with only a polynomial number of nonzero entries.
Such vectors can be efficiently represented by a list of the elements with positive values alongside their values.

Our solution utilizes a technique recently introduced in \cite{hartman2025reducing}. They show a generic reduction from leximin optimization to \emph{weighted-utilitarian optimization}, defined as follows (where $\alpha\in(0,1]$ is an approximation parameter)

\begin{tcolorbox}[left=2pt, right=2pt,  
%
%
colback=gray!5,                  
  colframe=gray!60,        
  colbacktitle=gray!20,           
  coltitle=black,                 
title=\textbf{(*) $\alpha$-Approximate Weighted-Utilitarian Allocation of Electricity}]
  \textbf{Input:} $n$ non-negative constants $v_1, \ldots, v_n$.
  \tcblower
  \textbf{Output:} A feasible configuration $C^{uo} \in \mathcal{C}$, for which:
  \begin{align*}
      \sum_{i =1}^n v_i\cdot u_i(C^{uo}) \geq \alpha \sum_{i =1}^n  v_i\cdot u_i(C) \quad \text{ for any } C \in \mathcal{C}
  \end{align*}
\end{tcolorbox}

Their main theorem is as follows.
\begin{theorem}[\cite{hartman2025reducing}]
\label{thm:hartman2025reducing}
    For any $\alpha\in(0,1]$, given a black box that solves 
    $\alpha$-Approximate Weighted-Utilitarian Allocation, 
    one can compute in polynomial time a lottery over configurations for which the vector of expected utilities is an $\alpha$-leximin approximation.
\end{theorem}
While in their framework solutions are lotteries over deterministic outcomes, in our setting, the same structure applies--—the resulting probabilities can be naturally interpreted as connection times (fractions of the interval $[0,1]$).
This correspondence enables us to apply their reduction directly.

Below, we provide an FPTAS for weighted-utilitarian allocation of electricity (*).
By \Cref{thm:hartman2025reducing}, this implies an FPTAS for leximin. 

To begin, observe that for any feasible configuration $C$, its weighted utilitarian value, $\sum_{i =1}^n  v_i\cdot u_i(C)$, equals $\sum_{i \in C} v_i$. Without geographic constraints, problem (*) is equivalent to the  Knapsack problem: each of the $n$ agents corresponds to an item, the demand $d_i$ corresponds to the size of item $i$, the total supply $S$ corresponds to the Knapsack capacity, and $v_i$ is the value of item $i$. This classic problem admits a well-known FPTAS \cite{vazirani2001approximation}.
The challenge arises when we add geographic constraints. In the remainder of this section, we address this more complex variant. Specifically, we show that when the connectivity graph is a tree, the Knapsack FPTAS can be adapted to this case.

\subsection{Geographic Knapsack}
This section introduces the \emph{Geographic Knapsack} problem. This section is self-contained, with its own notation and definitions.
Given a list of values $v_1,\ldots,v_n$ and a subset $C\subseteq [n]$, we denote $v_C := \sum_{i\in C} v_i = $ the total value of $C$. Similarly, given a list of sizes $d_1,\ldots,d_n$ we denote $d_C := \sum_{i\in C} d_i = $ the total size of $C$.

\begin{tcolorbox}[left=2pt, right=2pt, 
colback=gray!5,                  
  colframe=gray!60,        
  colbacktitle=gray!20,           
  coltitle=black,                 
title=\textbf{Geographic Knapsack}]
  \textbf{Input:~~~~} 
\begin{minipage}{0.85\linewidth}
\begin{itemize}
    \item $n$ items, each with its own size $d_i \in \mathbb{R}_{\geq 0}$ and value  $v_i \in \mathbb{R}_{\geq 0}$.

    \item Knapsack capacity $S$.

    \item A graph $G:= (V,E)$ with $V:=\{1,\ldots,n\} \cup \{s\}$
\end{itemize}
\end{minipage}\\[0.15em]
  \tcblower
  \textbf{Output:~~} 
  \begin{minipage}{0.85\linewidth}
  A \emph{legal packing} $C$ with maximum total value $v_C$, where a legal packing is a subset of the items, $C \subseteq N$, for which the following hold:
  \begin{itemize}
      \item Capacity Constraint: the total size does not exceed the knapsack capacity: $d_C  \leq S$.

      \item Geographic Constraint: the subgraph $G':=(V',E')$ where $V':=C \cup\{s\}$ and $E'$ the set of edges induced by $V'$, is connected. 
  \end{itemize}
  \end{minipage}\\[0.15em]
  
\end{tcolorbox}

The importance of this problem is demonstrated in the following theorem:
\begin{theorem}
    Let $\alpha \in (0,1]$. Given an $\alpha$-approximate solver for the Geographic knapsack problem, one can obtain an $\alpha$-leximin approximation for the electricity division problem in polynomial time.
\end{theorem}

\begin{proof}
    By \Cref{thm:hartman2025reducing}, in order to compute an $\alpha$-leximin-approximation for the electricity division problem, it is sufficient to provide an $\alpha$-approximate solver for the weighted utilitarian electricity allocation problem (*). 

    Given an input to the weighted utilitarian problem, we use the solver for the Geographic Knapsack problem as follows: each of the $n$ agents corresponds to an item, the demands, $d_1, \ldots, d_n$, correspond to the sizes of the items, the total supply $S $ corresponds to the knapsack capacity, and the given constants $v_1, \ldots, v_n$ correspond to the values of the items.

    Any legal packing in the geographic knapsack problem
    corresponds to a feasible configuration in the electricity problem.
    The total value of the packing equals the total weighted utility of that configuration.
    In particular, a legal packing is optimal if and only if its corresponding (legal) configuration is optimal, and it is an $\alpha$-approximation if and only if the configuration is an $\alpha$-approximation.
    \qed
\end{proof}


Before providing an FPTAS for Geographic Knapsack, we note that recently, and simultaneously to our work, \cite{dey2024knapsack} have developed an FPTAS for a different variant of Knapsack problem with connectivity constraints over a general graph. 
Their results are not directly applicable to our setting, as in our setting there is a predefined ``source'' node (representing the power station), and each selected subset must also be connected to this node. However, if their algorithm could be adapted to this setting, then it could be plugged into the reduction, and yield an FPTAS for leximin-optimal electricity distribution for any graph. 
We discuss this direction in the Future Work section.






Similarly to the known FPTAS for Knapsack \cite{vazirani2001approximation},\footnote{We recommend \href{https://www.youtube.com/watch?v=D372XEIP4-8&ab_channel=ComputerScienceTheoryExplained}{this video} by 
Computer Science Theory Explained for a clear explanation.} we propose an algorithm for the case where item values are integers. We apply the same value-rounding technique used in the standard algorithm, ensuring that an optimal packing for the rounded instance yields a $(1 - \varepsilon)$-approximation with respect to the original (unrounded) values.
For completeness, we present the rounding procedure and establish its correctness, although this technique is standard and not part of our contribution.

Let $\theta:=(\epsilon\cdot v_{max})/n$ where $v_{max} := \max_i v_i$. We round the values $v_i$ as follows: $r_i := \lfloor v_i / \theta\rfloor$.

\begin{lemma}
    \label{approx solution lemma}
    Let $C^{r} \subseteq N$ be a legal packing with maximum value w.r.t. the rounded values $r_i$. ~Let $C^*$ be a legal packing with maximum value w.r.t. the original values. 
    ~ Then:
    \begin{align*}
        \sum_{i \in C^r} v_i \geq (1-\epsilon) \sum_{i \in C^*} v_i.
    \end{align*}
\end{lemma}

\begin{proof}

    \begin{align*}
        \sum_{i \in C^r} v_i &= \sum_{i \in C^r} \theta \frac{ v_i}{\theta} \geq \sum_{i \in C^r} \theta \lfloor\frac{ v_i}{\theta} \rfloor = \sum_{i \in C^r} \theta \cdot r_i && \text{(By definition of $r_i$)}\\
        &  = \theta \sum_{i \in C^r}  r_i \geq \theta \sum_{i \in C^*}  r_i && \text{(as $C^r$ is optimal w.r.t. $(r_i)_i$)}\\
        &  = \theta \sum_{i \in C^*} \lfloor\frac{ v_i}{\theta} \rfloor  \geq \theta \sum_{i \in C^*}  \left(\frac{ v_i}{\theta} -1\right)  && \text{($\lfloor \cdot \rfloor$ reduces values by at most $1$)}\\
        & = \sum_{i \in C^*} v_i - \sum_{i \in C^*} \theta = \sum_{i \in C^*} v_i - n \frac{\epsilon\cdot v_{max}}{n} && \text{(By definition of $\theta$)}\\
        & = \sum_{i \in C^*} v_i - \epsilon\cdot v_{max} 
        \\
        & \geq \sum_{i \in C^*} v_i - \epsilon \sum_{i \in C^*} v_i && \text{(as $v_{max} \leq \sum_{i \in C^*} v_i$)}\\
        & = (1-\epsilon)  \sum_{i \in C^*} v_i.
    \end{align*}
\end{proof}
\Cref{approx solution lemma} implies that, to get an FPTAS for Geographic Knapsack, it is sufficient to find an exact algorithm for the special case in which all values are integers.
Our solution is given in \Cref{alg:SolveGeographicKnapsack}. It computes a legal packing of maximum value in the special case where the graph $G$ is a tree and all item values are integers.
The rest of this section is devoted to explaining this algorithm.

In our tree network, the root is labeled $s$. 
we renumber the nodes such that the following conditions hold:
\begin{enumerate}
    \item Nodes further from the root are given smaller numbers, that is, if $distance(s,i) > distance(s,i')$ then $i < i'$.
    \item Sibling nodes (nodes with the same parent) are given consecutive numbers.
\end{enumerate}
This renumbering can be done in time $O(|V|+|E|)$.

We denote by $c(i)$ the highest-numbered child of node $i$, and by $p(i)$ be the parent of node $i$.
Note that, by condition 2, if $p(i)\neq p(i-1)$, it means that $i$ is the lowest-numbered child of its parent; in that case we say that $i$ is a ``first child''. 
See Figure \ref{fig:tree} for an example.

\begin{figure}[h]
    \centering
\includegraphics[width=0.6\textwidth]{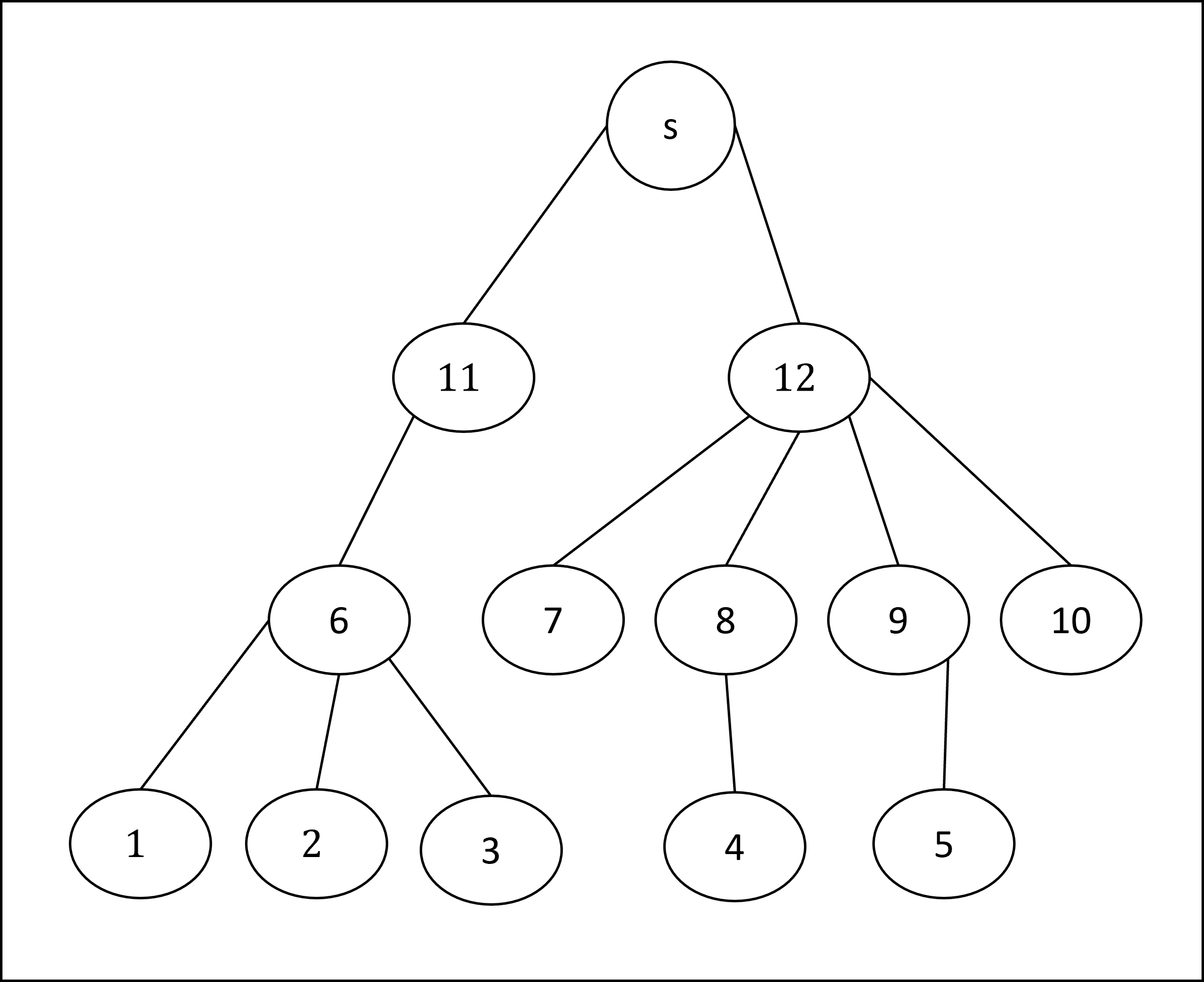}
\caption{An example graph after renumbering of the nodes. 
Note that the nodes farthest from the root are given the smallest numbers (1, 2, 3, 4, 5).
Also, all siblings (children of the same parent node) have consecutive numbers, e.g. the siblings $7,8,9,10$ of parent node $12$ are given consecutive numbers. 
Here, each of the nodes 1, 4, 5, 6, 7, 11 is a ``first child''.
Moreover, we have $c(6)=3, c(8)=4, c(9)=5, c(11)=6, c(12)=7$.
}
\label{fig:tree}
\end{figure}

Let $v_N := \sum_{i \in N} v_i$. 
The algorithm constructs a table $F$ of size $n\times (v_N+1)$, where $F(i,j)$ represents a subset $C\subseteq [i]$ such that $v_C = j$ and the following condition holds:
\begin{align}
    \label{eq:connected}
    C\cup \{\parentOf{i}\} \text{~induces a connected subgraph of $G$.}
\end{align}
Among all subsets satisfying these conditions, $F(i,j)$ contains a subset $C$ for which the total size $d_C$ is minimum. Formally:
\begin{align*}
    F(i,j) := \arg \min_{\{C\subseteq [i]:~
    v_C=j,~C\cup \{\parentOf{i}\} \text{~is~connected.}
    \}}
    d_C
\end{align*}
Informally, $F(i,j)$ is the most ``economic'' (= with smallest total size) subset of the agents $1,\ldots,i$, among the subsets with total value exactly $j$, that are connected through the parent of $i$.

Initially, all cells of the table $F$ are empty. The main part of \Cref{alg:SolveGeographicKnapsack} is devoted to filling the cells of $F$ in order.

\begin{algorithm}
    \caption{An Algorithm for Geographic Knapsack with Integer Values in Tree Graphs}
    \label{alg:SolveGeographicKnapsack}
\KwIn{Sizes $d_1,\ldots,d_n \in \mathbb{R}_{>0}$, integer values $v_1,\ldots,v_n \in \mathbb{N}$, and a tree graph $G$.}
    \SetAlgoLined
    \DontPrintSemicolon
    \LinesNumbered
    \SetKw{KwTo}{to}

     Initialize 
     $F(i,j) := None$; 
     for
     all $i\in N$ and all $j>0$.\;

    \Comment{/* Induction base: */}
    
    \For{$i=1$ \KwTo $n$}{
    
     $F(i,0) := \{\}$; 
     \;
     \Comment{/* The empty packing has value 0. */}

    \If{($i$ is a leaf node and $i$ is a first child)}
    {
    $F(i,v_i) = \{i\}$;
    \;
    \Comment{/* The only  packing of $[i]$ that satisfies \eqref{eq:connected} is $\{i\}$. */}
    }
    }

    \Comment{/* Induction step: */}
    
    \For{$i=1$ \KwTo $n$}{
        \For{$j=1$ \KwTo $v_N$}{
            \Comment{/*** Compute $C_{-i} := $ a  minimum-size legal packing without $i$. ***/}
            
            \Comment{/* Condition \eqref{eq:connected} implies that it must use some previous sibling of $i$. */}
            
            \eIf{$i$ is not a first child}
                {
                $C_{-i} := F(i-1,j);$ 
                \; 
                }
            {
                $C_{-i} := None$; 
                \; 
            }

            \Comment{/*** Compute $C_i := $  a  minimum-size legal packing with $i$. ***/}

            \uIf{$v_i > j$}
            {
            \Comment{/* We cannot add $i$ at all. */}
            
            $C_i := None$;
            }
            \uElseIf{$i$ is a leaf node and $i$ is not a first child} {
            \Comment{/* $i$ has previous siblings but no children. */}
            
                $C_i := F(i-1,j-v_i)\cup \{i\}$;
            }
            \uElseIf{$i$ is not a leaf node and $i$ is a first child}{
            \Comment{/* $i$ has children but no previous siblings; note that $c(i)<i$. */}
            
                $C_i := F(c(i), j-v_i)\cup \{i\}$;
            }
            \Else
            {
            \Comment{/* $i$ has both children and siblings; we have to divide the value among them */ }

            $\mathcal{F} := \emptyset$; \Comment{Initialization}
            
            \For{$z = 0$ \KwTo $(j - v_i)$}{
            \Comment{/* $z$ represents the value of $i$'s previous siblings. */}
            
            $F_s := F(i-1, z); F_c := F(c(i), j-v_i-z)$ \;
            $\mathcal{F} := \mathcal{F} \cup \{F_s \cup F_c \cup \{i\}$\}\;

            \Comment{/* Add to $\mathcal{F}$ the legal packing in which the value coming from the siblings is $z$.  */}
            } 
            $\displaystyle C_i := \arg \min_{C\in \mathcal{F}} d_{C}$;
            
            \Comment{/* $C_i$ has the smallest total size among all possible divisions of value */}
            }
            }
                $\displaystyle F(i,j) = \arg\min_{C \in \{C_i, C_{-i}\}} d_C$. \;
                \Comment{/* Take either $C_i$ or $C_{-i}$, whichever has a smaller total size */}
    }
    \For{$j = v_N$ \KwTo $1$}{
        \If{$F(n,j)$ is not None and $d_{F(n,j)} \leq S$}{
        \Comment{/* Found a legal packing with the highest total value */}
        
        \Return $F(n,j)$ and $j$ \;
        }
    }
\end{algorithm}

\begin{theorem}\label{thm:geographic-knapsack-integer-values}
    When $G$ is a tree and all item values, $v_1,\ldots,v_n$, are integers, \Cref{alg:SolveGeographicKnapsack} returns a legal packing of maximum value.

    If there are several such packings, the algorithm returns one with a smallest total size.
\end{theorem}

\begin{proof}
We prove that, when the algorithm terminates, each cell $F(i,j)$ contains a packing $C$ that satisfies $v_C = j$ and condition \ref{eq:connected} (if such a packing exists), and subject to that, the size $d_C$ is minimized.
The algorithm returns an answer from $F(n,j)$ with a largest $j$ such that $d_{F(n,j)} \leq S$.
Since $p(n) = s$, this would imply the correctness of the algorithm.

First, by checking all steps in which $F(i,j)$ is updated, one can prove (by induction on $i$) that, if $F(i,j)$ is not None, then it contains a packing with value exactly $j$.

It remains to prove that in each iteration the packing satisfies condition \ref{eq:connected}, and that its total size is minimized.

The proof is by induction on the order in which the cells in $F$ are filled, that is, we do induction on $i$, and for each $i$, we do induction on $j$.

\paragraph{Base cases.} The base cases are the cells for which $i$ is both a leaf and a first-child, 
and the cells for which $j=0$. 
They are initialized in the first For loop.

 The connectivity requirement \ref{eq:connected} is trivially satisfied, since the packings are either empty or singletons.
 For the value $j=0$, clearly the minimum-size packing is the empty packing. 
 For a first-child leaf, the only packing of $[i]$ that is connected to $p(i)$ is the singleton containing only $i$, so it is trivially the minimum-size one.
 Note that the cells filled in this case are not modified later, so their value remains correct
 
 \paragraph{Inductive step.}
        
We prove that, if all the previously-filled values in the table are correct, then the current value is correct as well.
Formally, if $F(i',j')$ is correct for all $(i',j')$ such that either ($i' < i$) or ($i' = i$ and $j' < j$), then $F(i,j)$ is correct too.
There are two cases to consider, depending on whether the packing in $F(i,j)$ contains $i$ or not.
The packing that contains $i$ is denoted $C_i$, and the packing without $i$ is denoted $C_{-i}$. 
We let $F(i,j)$ be either $C_i$ or $C_{-i}$, whichever has a smaller total size. 
Therefore, it is sufficient to prove the minimality of each of these two packings subject to the connectivity constraint \eqref{eq:connected}.

Suppose first that $F(i,j)$ is chosen to be $C_{-i}$, and consider two cases.
\begin{itemize}
    \item If $i$ is not a first child, and $i$ is not included in the packing, then the minimum-size legal packing of $[i]$ equals the minimum-size legal packing of $[i-1]$ with the same constraints and the same total value.
    Note that in this case $p(i-1) = p(i)$, so condition \eqref{eq:connected} is the same condition for $F(i-1,j)$ as for $F(i,j)$.
    Therefore, $C_{-i}$ equals $F(i-1,j)$, which is correct by the induction assumption.
    \item 
    If $i$ is a first child, and $i$ is not included in the packing, then the packing does not include any child of $p(i)$. However, condition \eqref{eq:connected} requires that the packing with $p(i)$ be connected. This means that no packing satisfying \eqref{eq:connected} exists, so $C_{-i} = None$.
\end{itemize}

Next, suppose that $F(i,j)$ is chosen to be $C_{i}$, and consider four cases.
\begin{itemize}
\item 
If $v_i >j$, then no packing containing $i$ can have a value of exactly $j$, so trivially $C_i = None$.

\item If $i$ is a leaf node and not a first child of its parent,
then the minimum-size subset of $[i]$ with value $j$ that includes $i$  must contain the minimum-size subset of $[i-1]$ with value $j' := j-v_i$ that is connected to $p(i)$.
By the induction assumption, 
and since $p(i) = p(i-1)$ 
(condition \eqref{eq:connected} is the same for $i-1$ as for $i$), 
that subset equals $F(i-1, j')$.
Hence, $C_i = F(i-1,j') \cup \{i\}$.
    
\item If $i$ is not a leaf but a first child of its parent, 
then the minimum-size subset of $[i]$ with value $j$ that includes $i$
must contain a minimum-size subset with value $j' := j-v_i$, that is connected to $i$.
That minimum-size subset is computed for the highest-indexed child of $i$, denoted by $c(i)$. The node numbering implies that $c(i)<i$. 
Note also that $p(c(i)) = i$, and $i$ is connected to $p(i)$.
Hence, by the induction assumption, that minimum-size subset equals $F(c(i), j')$,
so $C_i = F(c(i),j') \cup \{i\}$. Hence, $C_i$ satisfies the condition \ref{eq:connected} and it has smallest total size.

\item If $i$ has both children and siblings,
then the minimum-size subset that includes $i$ may contain both previous siblings of $i$ (with $i' < i$), and children of $i$ (with $i'\leq c(i) < i$).
We choose $C_i$ to be the best among all possible packings $F(i-1,j') \cup F(c(i), j'') \cup \{i\}$,
for all $j',j''$ such that $j'+j''+v_i = j$. By induction hypothesis 
$F(i-1,j') \cup \{p(i-1)\} = F(i-1,j') \cup \{p(i)\}$ is connected. 
Similarly, $F(c(i), j'') \cup p(c(i)) = F(i-1,j'') \cup \{i\}$ is also connected. 
Therefore, the union
$\{F(i-1,j') \cup F(c(i), j'') \cup \{i\} \cup \{p(i)\}\}$ is connected too (through $i$ and $p(i)$).
Hence $C_i$ satisfies condition \eqref{eq:connected}.
\end{itemize} 

Overall, we have shown that $F(i,j)$ satisfies  condition \ref{eq:connected} and it has smallest total size. This completes the induction step.

\qed
\end{proof}

\begin{theorem}
    \label{algorithm1 time complexity}
    Algorithm \ref{alg:SolveGeographicKnapsack} is a FPTAS for the Geographic Knapsack problem.
\end{theorem}
\begin{proof}
    By \Cref{thm:geographic-knapsack-integer-values} and \Cref{approx solution lemma}, the value of the returned packing is at least $(1-\epsilon)\cdot OPT$ where $OPT$ is the value of the optimal packing.

    It is left to prove that the run-time is polynomial in $n$ and $1/\epsilon$.

    We first note that the outer loop iterates from $1$ to $n$, and the inner loop iterates from $1$ to $v_N$. 
    
    The upper bound on $v_N$ is $n$ times the highest value. Since we use the rounded values $r_i:= \lfloor v_{i} / \theta \rfloor$, the highest values we used is $r_{max} := \lfloor v_{max} / \theta \rfloor$. 
    As $\theta=\epsilon \cdot v_{max} /n$, we get that $r_{max} = \lfloor n / \epsilon \rfloor$.
    Thus, $v_N = n \cdot r_{max}$ is at most $n \cdot r_{max} \leq n \cdot n / \epsilon = n^2\cdot 1/ \epsilon$. 
    
    In addition, each cell takes time polynomial in $v_N$ to compute. 
    Therefore, the total computation time is polynomial in $n$ and $1/\epsilon$.
    \qed
\end{proof}


\section{Conclusion and Future work}

This paper studies the problem of fairly distributing electricity among households when their collective demand exceeds the available supply. We focused on egalitarianism, which seeks to achieve fairness by prioritizing the utilities of the worst-off agents. Specifically, we examined the egalitarian rule, which selects an allocation that maximizes the minimum utility, and its generalization, the leximin rule, which further breaks ties by maximizing the second-lowest utility, then the third-lowest, and so on.

We prove that computing an egalitarian-optimal allocation is NP-hard even under simplified assumptions.
On a more positive note, we provide a Fully Polynomial-Time Approximation Scheme (FPTAS) for computing a leximin-optimal allocation when the connectivity graph is a tree. Our approach combines a reduction to a variant of the knapsack problem with Geographic constraints and an adaptation of existing approximation techniques.

This work opens up several directions for further research.

\paragraph{General Graphs.} Our current FPTAS relies on the assumption that the connectivity graph is a tree. An important open question is whether the algorithm can be extended to more general graphs. We conjecture that this could be done by adapting recent techniques developed by Dey, Kolay, and Singh \cite{dey2024knapsack}.

\paragraph{General utilities.} 
The algorithm we present assumes that agents have uniform and additive utilities. Without this assumption, we cannot directly apply the reduction of Hartman et al. \cite{hartman2025reducing}. Can we approximate leximin for broader classes of utilities, such as piecewise-uniform or piecewise-constant functions?

\paragraph{Contiguous Allocations.} In our model, agents may receive disconnected time intervals. An interesting direction, often considered in the cake-cutting literature, is to study fair allocation under the additional requirement that each agent receives a single contiguous interval.

\paragraph{Beyond Egalitarianism.} It would be interesting to explore other solution concepts, such as maximizing the Nash welfare (the product of utilities), 
as well as fairness notions such as proportionality, envy-freeness, or maximin share (MMS). What is the best we can do under these criteria?

\paragraph{Incorporating Priorities.}
In many applications, priorities play a key role in determining which demands should be served first. 
Our setting can be viewed through this lens: the geographic constraints naturally induce priority constraints, where nodes closer to the source represent more urgent or essential needs.
We believe this direction has wide applicability. For instance:
\begin{itemize}
    \item In smart home energy systems, appliances may be prioritized based on user preferences or criticality, leading to a tree-like dependency structure over time.
    
    \item In emergency response, resources must be allocated geographically, with preference given to regions closer to the source or to areas with higher urgency.
\end{itemize}
It would be interesting to extend our framework to handle other forms of priority constraints beyond graph-based ones.

\section{Acknowledgments}
This research is partly supported by the Israel Science Foundation grants 712/20, 3007/24, 2697/22 and 1092/24.

We sincerely appreciate Yonatan Aumann and Avinatan Hassidim for their valuable insights and discussion.

\bibliographystyle{splncs04}
\bibliography{references,references_erel,references-dinesh}
\end{document}

%% file: includes-and-macros-sagt2025.tex
\usepackage[utf8]{inputenc}

\usepackage{amsmath}
\usepackage{amssymb}

\usepackage{algcompatible}

\usepackage{mathtools}
\usepackage{mathrsfs}
\DeclareMathAlphabet{\mathpzc}{OT1}{pzc}{m}{it}

\usepackage[dvipsnames]{xcolor}

\usepackage{blindtext}
\usepackage{hyperref}

\usepackage{csquotes}
\usepackage{array}

\usepackage[normalem]{ulem}

\ifdefined\DRAFT

	\newcommand{\yon}[1]{\textcolor{Green}{(Yonatan says: #1)}}
     
\else

	\newcommand{\yon}[1]{}
\fi

\newcommand{\parentOf}[1]{p(#1)}